\PassOptionsToPackage{utf8}{inputenc}
\documentclass[noinfo,final,crop]{bioinfo}

\copyrightyear{2018}
\pubyear{2018}
\access{Advance Access Publication Date: Day Month Year}
\appnotes{Manuscript Category}

\usepackage{cite} 
\usepackage{url}  
\usepackage{booktabs}
\usepackage{graphicx} 
\usepackage{tikz}
\usetikzlibrary{arrows}

\usepackage{ifpdf}

\usepackage{amsthm}

\newtheorem{proposition}{Proposition}

\newcommand{\q}{\phantom0}
\newcommand{\qq}{\phantom{00}}

\def\mmod{~\textrm{mod}~}

\begin{document}
\firstpage{1}

\subtitle{Sequence Analysis}

\title[copMEM: Finding maximal exact matches...]{copMEM: Finding maximal exact matches via sampling both genomes}
\author[Szymon Grabowski and Wojciech Bieniecki]{Szymon Grabowski\,$^{\text{\sfb 1}}$ and Wojciech Bieniecki\,$^{\text{\sfb 1,}*}$}
\address{$^{\text{\sf 1}}$Institute of Applied Computer Science, Lodz University of Technology, Al.\ Politechniki 11, 90-924 {\L}\'{o}d\'{z}, Poland.}

\corresp{$^\ast$To whom correspondence should be addressed.}

\history{Received on XXXXX; revised on XXXXX; accepted on XXXXX}
\editor{Associate Editor: XXXXXXX}

\abstract{
\textbf{Summary:} 
Genome-to-genome comparisons require designating anchor points, which are given by Maximum Exact Matches (MEMs) between their sequences.
For large genomes this is a challenging problem and the performance of existing solutions, even in parallel regimes, is not quite satisfactory.
We present a new algorithm, copMEM, that allows to sparsely sample both input genomes, with sampling steps being coprime.
Despite being a single-threaded implementation, copMEM computes all MEMs of minimum length 100 between the human and mouse genomes in less than 2 minutes, using less than 10\,GB of RAM memory.\\
\textbf{Availability:} 
The source code of copMEM is freely available at 
\url{https://github.com/wbieniec/copmem}.\\
\textbf{Contact:} \href{wbieniec@kis.p.lodz.pl}{wbieniec@kis.p.lodz.pl}
}

\maketitle

\section{Introduction}
Maximal exact matches (MEMs) are exact matches between
two strings (genomes) that cannot be extended to the left or right 
without producing a mismatch. 
For high-throughput sequencing data, finding MEMs has two 
basic applications: 
$(i)$ seeding alignments of sequencing reads for genome assembly, and
$(ii)$ designating anchor points for genome-genome comparisons.

Early algorithms for MEM finding were based on a suffix tree~\citep{KPDSSAS2004} 
or an enhanced suffix array~\citep{AKO2004}. 
The space occupancy of these data structures instigated 
researchers to devise more succinct (and possibly also faster) solutions, 
including essaMEM~\citep{VDBFD2013} based on a sparse suffix array 
and E-MEM~\citep{KI2015}, which employs a hash table for sampled $K$-mers. 
E-MEM is the most succinct and also often the fastest solution.
Moreover, it allows to process the data in several passes, trading speed 
for a reduction in working memory.

Recently, Almutairy and Torng~(\citeyear{AT2018}) compared two approaches to 
sampling $K$-mers in order to find MEMs: fixed sampling and minimizer sampling.
The former, which is the approach of E-MEM, 
extracts $K$-mers from one of the input sequences in fixed  
sampling steps and inserts into a dictionary.
Then, all the $K$-mers from the other genome are compared against the seeds 
in the dictionary, in order to extend the matches.
The latter approach involves so-called minimizers,
and allows for sampling both genomes, but 
the number of shared $K$-mer occurrences is greater than with fixed sampling.
Although minimizer sampling may be slightly faster, it needs more space, 
hence the authors concluded that fixed sampling was the right way to go.

In this article, we propose a novel approach to sampling $K$-mers in both sequences, 
which combines speed and simplicity.
Our idea is based on
simple properties of coprime numbers.

\begin{table*}[ht]
\processtable{MEM results. Times in seconds, memory (RAM) usages in GBs (G = $10^{9}$).
\label{tab:mem}}
{
\renewcommand{\tabcolsep}{0.8em}
\centering
\begin{tabular}{lccccccccccccccc} \toprule
MEM alg.	&& \multicolumn{4}{c}{{\em H.\ sapiens} vs {\em M.\ musculus}} && \multicolumn{4}{c}{{\em H.\ sapiens} vs {\em P.\ troglodytes}} && \multicolumn{4}{c}{{\em T.\ aestivum} vs {\em T.\ durum}} \\ \cline{3-6} \cline{8-11} \cline{13-16} 
          && \multicolumn{2}{c}{$L=100$} & \multicolumn{2}{c}{$L=300$} && \multicolumn{2}{c}{$L=100$} & \multicolumn{2}{c}{$L=300$} && \multicolumn{2}{c}{$L=100$} & \multicolumn{2}{c}{$L=300$} \\
			  && Time	& RAM & Time	& RAM && Time	& RAM & Time	& RAM && Time & RAM & Time	& RAM  \\
\midrule
essaMEM -t 10	-k 4 &&   1779 &  14.03 &  1030 &  13.59 &&  3474 &  13.93 &  1544 &  13.97 &&  ---     & --- & --- & --- \\
eMEM -t 1   	     &&  \q976	& \q3.95 & \q575 & \q2.18 &&  2038 & \q4.10 & \q906 & \q2.34 && \q728 & \q5.65 & \q576 & \q3.17 \\
eMEM -t 10  	     &&  \q216	& \q4.03 & \q146 & \q2.26 && \q509 & \q4.18 & \q205 & \q2.41 && \q219 & \q5.72 & \q169 & \q3.24 \\
copMEM             &&  \q116 & \q9.70 & \qq55 & \q8.92 && \q382 &  10.92 & \qq93 & \q9.52 && \q223 &  17.88 & \q107 &  15.65 \\
\botrule
\end{tabular}
}
{Test platform: Intel i7-4930K 3.4\,GHz CPU, 64\,GB of DDR3-1333 RAM, SSD\,512 GB, Ubuntu 17.10 64-bit OS. 
All codes written in C++ and compiled with g++ 7.2.0 -march=native -O3 -funroll-loops.
I/O times are included. In the cases denoted with `---' essaMEM produced a multi-GB temp file and we had to stop it after more than an hour.}
\end{table*}

\begin{methods}

\section{Methods}
Given two sequences, $R_{0\ldots n_1 - 1}$ and $Q_{0\ldots n_2 - 1}$, 
the task is to find all MEMs of length at least $L$ symbols.
Following E-MEM,we sample $K$-mers from the reference 
genome $R$ with a fixed sampling step.
E-MEM sets the sampling step to $L - K + 1$ and inserts all sampled 
$K$-mers (seeds) in a hash table; the step choice guarantees that for any window 
of size $L$ one (and only one) $K$-mer will be sampled.
Then the query genome $Q$ is scanned, with the step equal to 1; 
all matching sampled $K$-mer occurrences from $R$ are found in the hash table  
and the tentative matches are left- and right-extended, 
to check if they are long enough.

Unlike E-MEM, our solution samples both genomes with step greater than 1.
To this end, we choose two positive integer parameters, $k_1$ and $k_2$, 
such that $gcd(k_1, k_2) = 1$ (where $gcd$ stands for the 
greatest common divisor) and $k_1 \times k_2 \leq L - K + 1$.
The $K$-mers from genome $R$ are extracted with step $k_1$ and inserted in a hash table.
The genome $Q$ is scanned with step $k_2$ and similarly its $K$-mers are extracted, 
candidate seeds are found in the hash table and then left- and right-extended.
As the key mechanism is based on coprimality of the parameters $k_1$ and $k_2$, 
we denote our algorithm as copMEM.

The correctness of our procedure implies from the following proposition.

\begin{proposition}
Let $k_1$ and $k_2$ be two positive integers that are coprime.
Let $r_1 \in \{0, 1, \ldots, k_1 - 1\}$. 
For any $r \in \{0, 1, \ldots, k_2 - 1\}$, 
the set 
$\mathcal{S} = \{i k_1 + r_1~|~i = 0, 1, \ldots, k_2 - 1\}$
contains an element of the form $j k_2 + r$, for some $j \geq 0$.
\end{proposition}

\begin{proof}
Let us create 
$\mathcal{S'} = \{(i k_1 + r_1) \mmod k_2~|~i = 0, 1, \ldots, k_2 - 1\}$.
All elements of $\mathcal{S'}$ are distinct.
Indeed, were it not the case, we would have $k_2 | v k_1 - u k_1$ 
for some $u$, $v$, such that $0 \leq u < v < k_2$,
which, in light of the coprimality of $k_1$ and $k_2$, implies that 
$k_2 | v - u$, a contradiction.
As the size of $\mathcal{S'}$ is $k_2$, all
$r$, $0 \leq r < k_2$, occur in it.
\qed
\end{proof}

Let $R_{i \ldots i+L'-1}$ and $Q_{j \ldots j+L'-1}$, where $L' \geq L$, 
form a MEM.
Denote $r_1 = i \mmod k_1$ and $r_2 = j \mmod k_2$.
If $R_{i \ldots i+L'-1}$ fully contains at least $k_2$ $K$-mers sampled with step $k_1$,
then, by the proposition above, 
for (at least) one of them, $R_{i' \ldots i'+K-1}$, we have $i' \mmod k_2 = r_2$.
As $R_{i \ldots i+L'-1}$ and $Q_{j \ldots j+L'-1}$ are equal, 
and all $K$-mers sampled from $Q_{j \ldots j+L'-1}$ start at position $j'$ such 
that $j' \mmod k_2 = r_2$, the $K$-mer $Q_{j+i-i' \ldots j+i-i'+L'-1}$ is among 
the sampled $K$-mers and it is equal to $R_{i' \ldots i'+K-1}$.
We thus showed that our sampling scheme cannot miss a seed for a matching window, 
provided appropriate choice of the parameters.
To this end, and to minimize the number of sampled positions (assuming also that 
the input genomes are of similar length, which is usually the case), 
we set $k_1 = \lceil \sqrt{L - K + 1} \rceil$ and $k_2 = k_1 - 1$.
%



\section{Results}
To evaluate the performance of copMEM, we chose the same large real datasets 
as used in the E-MEM paper.
The datasets are in multi-FASTA format, with sizes ranging from 2.7\,Gbs to 4.5\,Gbs.
Supplementary Material contains the dataset URLs and characteristics, 
as well as details on the test methodology.
MEM finding times and RAM usages are given in Table~\ref{tab:mem}.

The parameter -t for E-MEM and essaMEM is the number of threads.
We can see that despite the fact that copMEM in the current implementation 
is single-threaded, it is usually much faster than its competitors 
(running them with one thread yields close to an order of magnitude difference). 
In memory use, E-MEM is more frugal.
The amount of space that copMEM occupies in the RAM memory 
is roughly described (in bytes) by the following formula:
$|R| + |Q| + 4 \times (2^{29} + |R| / k_1) + |output|$, 
where $2^{29}$ is the number of slots in the hash table.


\end{methods}

\section*{Acknowledgement}

\paragraph{Funding\textcolon}
This work was financed by the Lodz University of Technology, 
Faculty of Electrical, Electronic, Computer and Control Engineering 
as a part of statutory activity (project no.\ 501/12-24-1-5418).



\clearpage
\onecolumn

\begin{center}
{\Large Supplementary Material to:\\
``copMEM: Finding maximal exact matches via sampling both genomes''\\
by\\
Szymon Grabowski and Wojciech Bieniecki}
\end{center}

\section*{Used datasets}

Please download the following datasets and extract them (each to single directory). 
If an archive contains multiple files, they have to be concatenated to be used 
as one of the two input files for copMEM (or other MEM tools in our test procedure).

\begin{sloppypar}
\begin{enumerate}
\item Homo sapiens\\
\url{ftp://hgdownload.cse.ucsc.edu/goldenPath/hg19/bigZips/chromFa.tar.gz} \\
(ca. 900\,MB of gzipped size)\\
\item Mus musculus\\
\url{ftp://hgdownload.cse.ucsc.edu/goldenPath/mm10/bigZips/chromFa.tar.gz} \\
(ca. 830\,MB of gzipped size)\\
\item Pan troglodytes\\
\url{ftp://hgdownload.cse.ucsc.edu/goldenPath/panTro3/bigZips/panTro3.fa.gz} \\
(ca. 900\,MB of gzipped size)\\
\item Triticum aestivum\\
\url{ftp://ftp.ensemblgenomes.org/pub/plants/release-22/fasta/triticum_aestivum/dna/Triticum_aestivum.IWGSP1.22.dna.genome.fa.gz} \\
(ca. 1.3\,GB of gzipped size)\\
\item Triticum durum\\
\url{https://urgi.versailles.inra.fr/download/iwgsc/TGAC_WGS_assemblies_of_other_wheat_species/TGAC_WGS_durum_v1.fasta.gz} \\
(ca. 970\,MB of gzipped size)\\
\end{enumerate}
\end{sloppypar}

Basic characteristics of the datasets are presented in Table~\ref{table:datasets}.

\section*{How to run copMEM}

Assuming that the input genomes are named human.all.fa and mus.all.fa, where the former is the reference and the latter the query genome,
an exemplary command line may look like:
\begin{verbatim}
./copmem -o hm.mems -l 300 human.all.fa mus.all.fa >hm_times.txt
\end{verbatim}

You can also use the switch \texttt{-v} (verbose), which prints more details.

\section*{Experiments}

The tests were conducted on two machines.
One was equipped with an Intel i7-4930K 3.4\,GHz (max turbo frequency 3.9\,GHz) CPU, 
64\,GB of DDR3-1333 RAM, SSD\,512 GB (Samsung 840 Pro), plus two HDDs,
and was running Linux Ubuntu 17.10 64-bit OS.
The sources were compiled with \texttt{g++ 7.2.0 -march=native -O3 -funroll-loops}.

The other was a 64-bit Windows 10 machine, 
with an Intel i5-6600 3.3\,GHz (max turbo frequency 3.9\,GHz) CPU, 
32\,GB of DDR4-2133 RAM and a fast 512\,GB SSD disk (Samsung 950 PRO NVMe M.2).
This computer was used only for running copMEM,
and the sources were compiled with Microsoft Visual Studio 2017, 
in the 64-bit release mode and using the following switches: /O2 /arch:AVX2.

The presented timings in the main part of the paper, 
as well as those given in Table~\ref{table:hash} below,
are median times of 3 runs in a particular setting 
(pair of datasets, chosen $L$ and possibly the number of threads).
Before each run, the memory of the Linux machine was filled with some dummy value, 
to flush disk caches.
In other words, all tests were run with a ``cold cache''.

\begin{table}[!ht]
\caption{Datasets used in the experiments.}
\vspace{2mm}
\label{table:datasets}
\centering
\setlength{\tabcolsep}{1.6em}
\begin{tabular}{lrr}
\hline
Dataset  & Size (MB) & Sequences \\
\hline
{\em Homo sapiens} (human) & 3,137 & 93 \\
{\em Mus musculus} (mouse) & 2,731 & 66 \\
{\em Pan troglodytes} (chimp) & 3,218 & 24,132 \\
{\em Triticum aestivum} (common wheat) & 4,391 & 731,921 \\
{\em Triticum durum} (durum wheat) & 3,229 & 5,671,204 \\
\hline
\end{tabular}
\end{table}

\smallskip

\begin{table}[!ht]
\caption{The impact of the used hash function. Times given in seconds. Experiments were run on the Windows machine.}
\vspace{2mm}
\label{table:hash}
\centering
\setlength{\tabcolsep}{1.6em}
\begin{tabular}{lrrrrr}
\hline
Task                                     &      default &     xx32 &     xx64 &     Metro64 &    City64 \\
\hline
{\em H.sap} vs. {\em M.mus}, $L = 100$   &        70.65 &    79.90 &    82.13 &       80.20 &      72.39 \\
{\em H.sap} vs. {\em M.mus}, $L = 300$   &        32.18 &    36.65 &    37.65 &       38.09 &      34.06 \\ 
{\em H.sap} vs. {\em P.tro},   $L = 100$ &       299.35 &   305.22 &   304.88 &      300.91 &     291.96 \\
{\em H.sap} vs. {\em P.tro},   $L = 300$ &        62.00 &    65.64 &    66.39 &       65.61 &      62.16 \\
{\em T.aest} vs {\em T.durum}, $L = 100$ &       151.30 &   172.70 &   176.37 &      167.96 &     156.31 \\ 
{\em T.aest} vs {\em T.durum}, $L = 300$ &        75.46 &    83.52 &    85.28 &       83.99 &      75.89 \\
\hline
\end{tabular}
\end{table}

Table~\ref{table:hash} presents the impact of the chosen hash function on the overall performance 
of copMEM.
Of course, the memory usage is unaffected, yet the speeds vary somewhat.
The tested hash functions include:

\begin{itemize}
\item xxHash by Yann Collet (\url{https://github.com/Cyan4973/xxHash}), both in 32-bit and 64-bit version,
\item MetroHash64 by J. Andrew Rogers (\url{https://github.com/jandrewrogers/MetroHash}),
\item CityHash64 from Google (\url{https://github.com/google/cityhash}).
\end{itemize}

The default hash function (used in the ``final'' copMEM version), is taken from \url{http://www.amsoftware.narod.ru/algo2.html}
(with slight changes in the code) and denoted there as maRushPrime1Hash.

\section*{Final notes}

Apart from the main algorithmic idea, several factors allowed us to achieve 
high performance.
It is important to use a large enough value of $K$ (set to 44 in copMEM)
and an efficient hash function.
A standard hash table (with chaining), as being a dynamic data structure, 
is not a good choice, due to millions of small data allocations.
Instead, we scan $R$ in two passes, first counting the number of occurrences of each hash 
value and then writing the sampled positions into appropriate locations of an array. 
This strategy, resembling the textbook counting sort, 
avoids memory fragmentation.
Finally, in a few places we applied software prefetching and other low-level optimizations.

\end{document}